\documentclass[runningheads]{llncs}
\usepackage{graphicx}
\usepackage{graphics}
\usepackage{array}
\usepackage{amsmath}
\usepackage{mathrsfs}
\usepackage{amssymb}
\usepackage{comment}
\usepackage{setspace}
\usepackage{algorithm}
\usepackage{algcompatible}
\usepackage{fullpage}
\usepackage{epstopdf}
\date{}
\newtheorem{obs}{Observation}
\title{Listing All Spanning Trees in Halin Graphs - Sequential and Parallel view}
\author{K.Krishna Mohan Reddy, P.Renjith, N.Sadagopan} 
\institute{Indian Institute of Information Technology, Design and Manufacturing, Kancheepuram, India. \\
\email{\{coe11b014,coe14d002,sadagopan\}@iiitdm.ac.in}}
\begin{document}
\maketitle
\setlength{\intextsep}{6pt}
\begin{abstract}
For a connected labelled graph $G$, a {\em spanning tree} $T$ is a connected and an acyclic subgraph that spans all vertices of $G$.  In this paper, we consider a classical combinatorial problem which is to list all spanning trees of $G$.  A Halin graph is a graph obtained from a tree with no degree two vertices and by joining all leaves with a cycle.  We present a sequential and parallel algorithm to enumerate all spanning trees in Halin graphs.  Our approach enumerates without repetitions and we make use of $O((2pd)^{p})$ processors for parallel algorithmics, where $d$ and $p$ are the depth, the number of leaves, respectively, of the Halin graph.  We also prove that the number of spanning trees in Halin graphs is $O((2pd)^{p})$.
\end{abstract}
\section{Introduction}
Enumeration of sets satisfying a specific property is an important combinatorial problem in the field of combinatorics and computing.  Popular ones are listing all spanning trees \cite{gabow,Ramesh,Matsui,Kapoor}, listing all minimal vertex separators \cite{Kloks}, enumerating maximal independent sets \cite{Shuji}, etc.   In \cite{Minty},  G.J.Minty  initiated the study of enumeration of spanning trees in general graphs as it finds applications in circuits and network systems.  Subsequently,  Read et al. \cite{read} and Shioura et al. \cite{Shioura} presented an algorithm to list all spanning trees in arbitrary graphs.  So far, algorithms for enumeration of spanning trees reported in the literature either follow backtracking approach or enumeration with the help of fundamental cycles.  Enumeration of trees with some structural constraints is reported in \cite{yamanaka,nakano}\\\\
As far as the bounds are concerned, Cayley \cite{Cayley} established that there are at most $n^{n-2}$ spanning trees on an $n$-vertex labelled graph and it is tight if the graph is complete.  Moon et al. have shown that there are $3^{\frac{n}{3}}$ maximal independent sets \cite{Meir} and T.Kloks et al. have established that there are $O(3^{\frac{n}{3}})$ minimal vertex separators \cite{Kloks}.  From the computing front, a natural question is to perform enumeration efficiently.  Since there are exponential number of feasible solutions in general graphs, any algorithm requires exponential effort to list all of them.  It is important to highlight the fact that since general graphs do not have nice combinatorial structure unlike special graph classes, some of the feasible solutions may be generated more than once during the enumeration process.  Therefore, a related problem is to perform enumeration without repetitions.\\\\
Since the number of spanning trees is exponential in the input size, any sequential algorithm incurs exponential effort to list all of them.  To speed up the enumeration, a natural alternative is to generate many feasible solutions in parallel.  Since modern day computers are equipped with multi-core processors, design of parallel algorithms not only speeds up the enumeration but also utilizes the underlying hardware resources effectively.  Having highlighted the importance of parallel algorithmics, in this paper, we shall investigate the enumeration of spanning trees in Halin graphs both from sequential and parallel perspectives. \\\\
Halin graphs (due to Halin \cite{Proskurowski}) are constructed from a tree with no degree two vertices and by joining all leaves  with a cycle.  
The objective of this paper is three folds.  The first one is to identify a graph class where enumeration of spanning trees can be done without repetitions.  The second one is to bound the number of spanning trees as a function of structural parameters rather than the input size.  The final one is to discover parallel algorithmics for enumeration so that many feasible solutions can be generated simultaneously.  It is important to note that Halin graphs have nice structural properties like $3$-connected, planar, Hamiltonian, and it is one of the popular subgraph of planar graphs which has been a candidate graph class for many classical problems such as Maximum Leaf Spanning tree, Steiner tree, etc.  Due to its planarity structure, it finds applications in VLSI design and computer networks \cite{Proskurowski}, and it has been an active graph class in the literature, see \cite{lou,lsc,mcf} for recent works on Halin graphs.
  Although combinatorial problems such as  planarity testing \cite{TarjanEPT}, bipartiteness testing \cite{Eppstein}, chordality testing \cite{Tarjan1985}, connectivity augmentation \cite{Frank,TarjanVishkin,chen}, etc., have received attention in parallel algorithmics, enumeration of sets satisfying some structural property have not received much attention in the past. \\\\
   To our best knowledge, this paper makes the first attempt in generating all spanning trees in Halin graphs without repetitions.  We exploit the combinatorial structure of Halin graphs in great detail and present a sequential and parallel algorithm for listing all spanning trees without repetitions.  Enumeration without repetition is achieved due to a nice structure of Halin graphs which is unlikely to exist in general graphs.   Our sequential algorithm uses a coloring technique to color some of the edges which are likely to create repetitions during enumeration process.  The overall structure of sequential algorithm naturally yields a parallel algorithm.  We also establish a bound on the number of spanning trees in Halin graphs which is $O((2pd)^{p})$ and this helps to fix the number of processors for parallel algorithmics.  Our sequential approach is incremental in nature and incurs a polynomial-time delay between successive spanning trees. \\\\
Combinatorial problems such as Hamiltonicity, Max-leaf spanning tree have polynomial-time algorithms when the input is restricted to Halin graphs \cite{lou} which are NP-complete in general graphs.  Since Halin graphs possess nice structural properties, it has been a candidate graph class to understand the gap between NP-completeness and polynomial-time solvability for combinatorial problems which are NP-complete on planar graphs.    In this paper, we exploit the combinatorial structure and perform enumeration without repetitions.  We believe that this is a major contribution as other enumeration algorithms reported in the literature does enumerate with repetitions. 
\subsection{Graph preliminaries}
For notations and definitions we follow \cite{west,golumbic}.  Let $G(V,E)$ be an undirected connected graph where $V(G)$ is the set of vertices and $E(G) \subseteq \{uv:u,v \in V(G)$, $u\neq v\}$.  \emph{Neighborhood} of a vertex, $N_{G}(w)=\{x:wx\in E(G)\}$ and $d_{G}(w)=|N_{G}(w)|$.  A vertex $w$ in a tree $T$ is a \emph{leaf} vertex if $d_{T}(w)=1$.  A Halin graph $H=T\cup C$ is constructed using a tree $T$ with no verex of degree two, and by connecting all leaves with a cycle $C$.  $T$ is termed as \emph{characteristic tree (Base tree)} and $C$ is termed as the \emph{accompanying cycle}.   Assuming $T$ is rooted, the \emph{depth} $d$ of $H$ is the depth of $T$ which is the length of the longest path from root to a leaf in $T$.  A Halin graph, its characteristic tree, and accompanying cycle are shown in Figure \ref{Halineg}.  For a Halin graph $H$, $C=(e_1,\ldots, e_p)$, $p\geq 3$ denotes an ordering of $p$ edges in the accompanying cycle $C$ such that for every $1\leq i\leq p-1$, $e_i=v_iv_{i+1}$ and $e_{p}=v_1v_p$.  $P_{uv}$ represents a path from $u$ to $v$.  We sometimes use $P_{uv}$ to represent $V(P_{uv})$, if the context is unambiguous.
Spanning trees $T_1$ and $T_2$ of a graph $G$ are said to be equal, if $E(T_1)=E(T_2)$ and denoted as $T_1=T_2$.  Graph $H$ is said to be an induced subgraph of $G$ if $V(H)\subseteq V(G)$, $E(H)=\{uv:u,v\in V(H)$ and $uv\in E(G)\}$ and is denoted as $H\sqsubseteq G$.
\begin{figure}[!h]
\begin{center}
\includegraphics[scale=1.5]{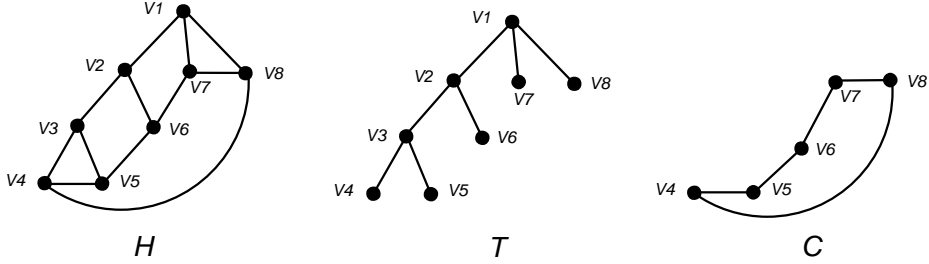}  
\caption{Halin Graph} \label{Halineg}
\end{center}
\end{figure}
\subsection{Parallel-algorithmic preliminaries}
In this paper, we work with Parallel Random Access Machine (PRAM) Model \cite{jaja}.  It consists of a set of $n$ processors all connected to a shared memory.  The time complexity of a parallel algorithm is measured using the processor-time product $($number of processors $\times$ time for each processor$)$.  Access policy must be enforced when two processors are trying to Read/Write into a cell. This can be resolved using one of the following strategies:
\begin{itemize}
\item Exclusive Read and Exclusive Write (EREW): Only one processor is allowed to read/write into a cell
\item Concurrent Read and Exclusive Write (CREW): More than one processor can read a cell but only one is allowed to write at a time
\item Concurrent Read and Concurrent Write (CRCW): All processors can read and write into a cell at a time.
\end{itemize}
In our work, we restrict our attention to CRCW PRAM model.  For a problem $Q$ with input size $N$ and $p$ processors,  the speed-up is defined as $S_p(N)= \frac{T_1(N)}{T_p(N)}$, where $T_p(N)$ is the time taken by the parallel algorithm on a problem size $N$ with $p$ ($p\geq 2$) processors and $T_1(N)$ is the time taken by the fastest sequential algorithm (in this case $p=1$) to solve $Q$.  The efficiency is defined as $E_p(N)=\frac{S_p(N)}{p}$.
\section{Listing all spanning trees in Halin graphs: A sequential approach}
In this section, we shall first present a sequential algorithm to enumerate all spanning trees in Halin graphs.  The sequential algorithm presented here is iterative in nature, and with the help of the base tree $T$ and the accompanying cycle $C$, we systematically generate all spanning trees which are stored in the set $ENUM$.  Further, we also present a bound on the number of spanning trees in Halin graphs, using its structural parameters. 
\subsection{Enumeration Algorithm}
The algorithm is simple, which starts the enumeration with the base tree $T$, further the algorithm iteratively adds an edge in $C$ to $T$, which creates a cycle $C^*$.  Spanning trees are enumerated by removing the edges in $E(C^*)\backslash E(C)$ one at a time to obtain other spanning trees.  
\begin{algorithm} 
\caption{Sequential algorithm to list all spanning trees of a Halin graph\newline {\em sequential-list-spanning-trees(H)}} \label{alg1}
\begin{algorithmic}[1]
\STATEx{{\tt Input:} A Halin Graph $H$}
\STATEx{{\tt Output:} All spanning trees of $H$. }
\STATEx{{\tt /* The set $ENUM$ contains all spanning trees of $H$, and $T$ is the characteristic tree of $H$. */ }}
\STATE{ Initialize $ENUM=\{T\}$.}
\STATE{$ \sigma $=$(e_1,\ldots,e_p)$  be an ordering of edges in C.}
\FOR{$i$ $=$ $1$ to $p$}
\STATE{  {\tt sequential-recursive-list(T,$e_i$)}.} \label{alg2call}
\ENDFOR
\end{algorithmic}
\end{algorithm}   
\begin{algorithm}
\caption{Sequential recursive listing of spanning trees in Halin graphs
\newline {\em sequential-recursive-list($T^{'}$,$e_i$)}}  \label{alg2}
\begin{algorithmic}[1]
\STATE{$G^*\leftarrow$ $T^{'}$+ $e_i$.}  
\STATE{Let $C^{*}$ be the unique cycle in $G^{*}$ containing $e_i$ .}
\STATE{$ \sigma^{*} $=$(b_1,\ldots,b_k)$  be an ordering of edges in $E(C^{*})\backslash E(C)$.}
\FOR{$m$ $=$ $1$ to $k$}
\STATE{$T^{*}\leftarrow $ $G^{*}-$ $b_m$.  Update	ENUM = ENUM $\cup$ $T^{*}$.} \label{addTei_enum}
\FOR{$j$ $=$ $i+1$ to $p$} 
\STATE{ {\tt sequential-recursive-list($T^{*}$,$e_j$)}.}
\ENDFOR
\ENDFOR
\end{algorithmic}
\end{algorithm}
\begin{figure}[!thb]
\begin{center}
\includegraphics[scale=0.95]{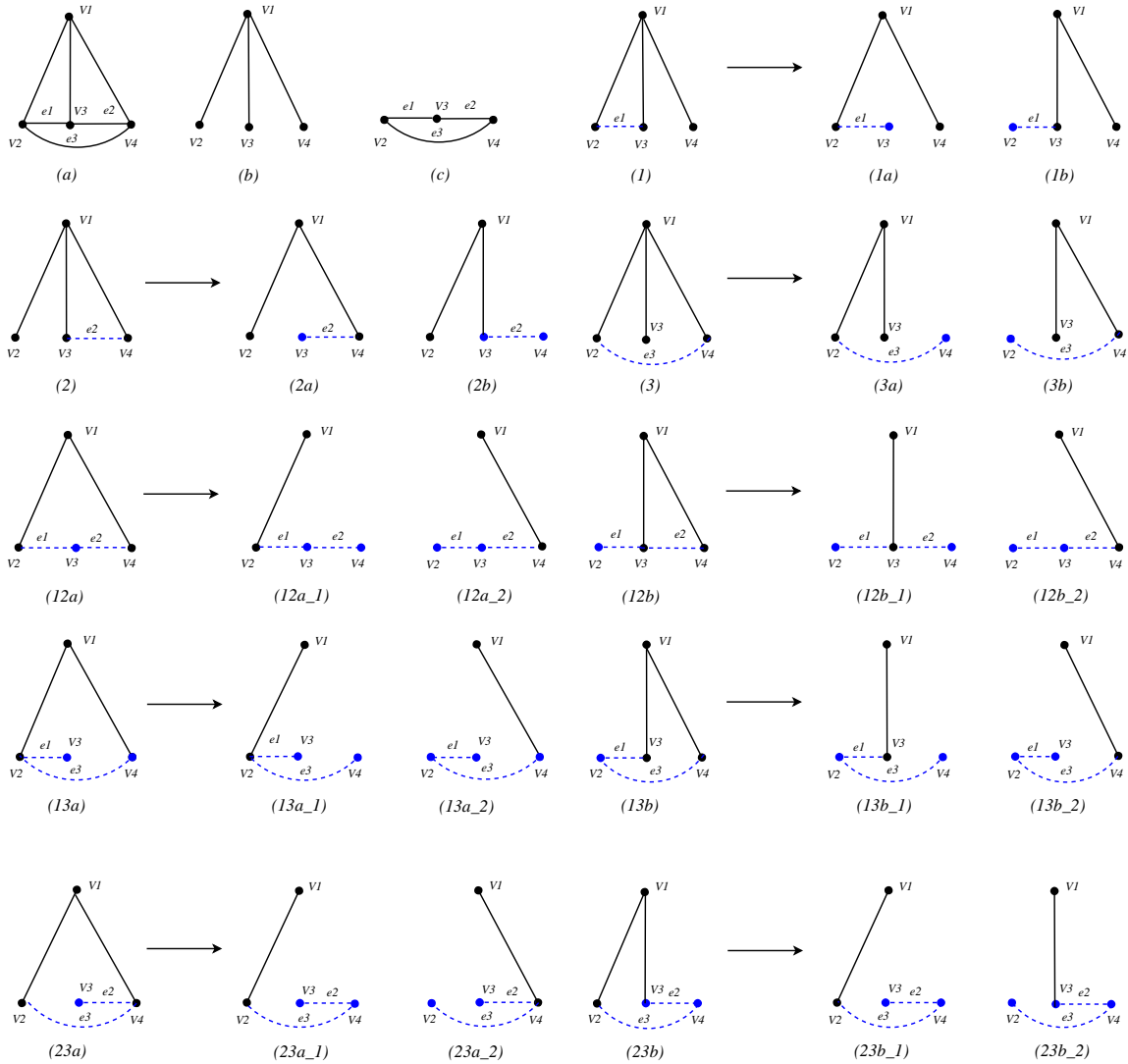}  
\caption{Trace Of Algorithm 1} \label{trace}
\end{center}
\end{figure}
\subsection{ Trace of Algorithm 1}
Figure \ref{trace}$.(a)-(c)$ shows a Halin Graph $H$ on 4 vertices, its characteristic tree $T$ and the accompanying cycle $C$, respectively.  Step $4$ of Algorithm $1$ calls Algorithm $2$ with parameters $(T$,$e_i),1\leq i\leq 3$ which adds edges $e_1,e_2,$ and $e_3$ to $T$ and yields new spanning trees shown in sub Figures $(1a)$-$(1b)$, $(2a)$-$(2b)$, and $(3a)$-$(3b)$, respectively.  Note that in Figure \ref{trace}, sub Figures $(1a)$ and $(1b)$ denote spanning trees containing $e_1$, which we denote using the set $T_{e_1}$.  Algorithm $2$ recursively adds edges $e_2$ and $e_3$ to each spanning tree in $T_{e_1}$ and generate the sets $T_{e_1,e_2}$, and $T_{e_1,e_3}$, respectively.   $T_{e_1,e_2}$ is the set of spanning trees containing $e_1$ and $e_2$ which are shown in sub Figures (12a\_1), (12a\_2), (12b\_1), (12b\_2).  Similarly, $T_{e_1,e_3}$ and $T_{e_2,e_3}$ are illustrated in Figure \ref{trace}.  Note that the spanning trees generated by the algorithm are not unique as (12a\_2) and (12b\_2) are identical copies of the same spanning tree.
%
\subsection{Proof of Correctness}
\begin{theorem}
For a Halin Graph $H$, Algorithm 1 enumerates all spanning trees of $H$.
\end{theorem}
\begin{proof}
Let $T^{'}$ be an arbitrary spanning tree of $H$.  We show that $T^{'}$ is generated by Algorithm \ref{alg1}.  If $T^{'}$ is the characteristic tree of $H$, then we are done.  If $E(T^{'})\cap E(C)\neq\emptyset$, then we show by induction on $n=|E(T^{'})\cap E(C)|,n\geq1$ that $T^{'}$ is generated by Algorithm \ref{alg1}.  \\
\emph{Base case:} $n=1$.  Our algorithm adds the edge $e_i$, $1\leq i\leq p$ to $T$, which creates a cycle $C^*$.  The algorithm then removes every edge in $C^*$ except $e_i$ of $C$.  This enumerates all spanning trees containing the edge $e_i$, $1\leq i\leq p$.  Therefore, Level 1 of the computational tree has all spanning trees having exactly one accompanying cycle edge.  Thus, our claim is true for the base case. \\
\emph{Induction Hypothesis:} Let us assume that Algorithm \ref{alg1} generates all spanning trees with less than $n$,$n\geq 2$ accompanying cycle edges.  That is, for every $i<n$, in the computational tree, all spanning trees containing exactly $i$ accompanying cycle edges are generated by our algorithm in Level $i$.\\
\emph{Induction Step:} Let $T^{'}$ be an arbitrary spanning tree on $n\geq 2$ accompanying cycle edges such that $E(T^{'})\cap E(C)$=$\{e_i,e_j,\ldots,e_k,e_m\}$  where $i<j<\ldots<k<m$.  By our Induction Hypothesis, the set $T_{e_i,e_j,\ldots,e_k}$ has less than $n, n\geq 2$ accompanying cycle edges.  For each spanning tree in $T_{e_i,e_j,\ldots,e_k}$, Algorithm 2 adds $e_{f}$, $k+1\leq f\leq p$ and generates new spanning trees using the cycles created due to this addition.  Clearly, in this process, Algorithm 2 adds $e_m$, $m>k$ and generate $T^{'}$.  The induction is complete and therefore, the theorem. \hfill \qed
\end{proof}
\subsection{Run-time Analysis}
Let $t$ be the total number of spanning trees possible for a given Halin graph.  Algorithm 1 adds the accompanying cycle edges one after the other to the characteristic tree of $H$.  After adding accompanying cycle edges one after another, to a given spanning tree, Algorithm 2 incurs, $O(n)$ time to detect the cycle formed, using the standard Breadth First Search algorithm.  Ordering of edges in $E(C^{*})\backslash E(C)$, $\sigma^*$ can be found in linear time, from which each edge can be removed in constant time to output a spanning tree.  So the total time taken by our sequential algorithm is $O(n t)$ which is $O(n(2pd)^{p})$.  Also notice that the algorithm incurs linear time delay between generation of two successive spanning trees.
\subsection{A bound on the depth of the Characteristic Tree}
\begin{theorem}
Let $H$ be a Halin Graph on $n$ vertices with $d$ being the depth of the characteristic tree of $H$.  Then, $d\leq \lfloor\frac{n}{2}\rfloor-1$.
\end{theorem}
\begin{proof}
We prove by induction on $d$, the depth of the characteristic tree $T$.\\  {\em Base case:} $d=1$. Clearly when $d=1$, the degree of the root vertex $d_{T}(v)\geq 3$.  Hence, $n\geq 4$ and $d\leq \lfloor\frac{n}{2}\rfloor-1$.  \\ {\em Induction Hypothesis:} Assume that every Halin graph on $n$-vertices with the depth $d\geq 1$ of its characteristic tree $T$ has $d\leq \lfloor\frac{n}{2}\rfloor-1$. \\ {\em Induction Step:} Consider a Halin graph $H=T\cup C$ with the depth of its characteristic tree $d$, $d\geq 2$.  Let $T^{*}$ be the tree obtained from $T$ by removing all leaves  at depth $d$ of $T$ and $H^{*}$ be the Halin graph obtained from $T^{*}$ by joining all leaves of $T^{*}$ with a cycle.  Now, the parent nodes of the removed leaf nodes become leaf nodes in $T^{*}$.  Let the number of vertices removed be $k$.  Clearly, $k\geq 2$ as for every internal vertex $v\in V(T)$, $d_{T}(v)\geq 3$.   Observe that, $T^{*}$ is the characteristic tree of the Halin graph $H^{*}$ with depth $d-1$, $d\geq 2$.  From the induction hypothesis, we have $d-1\leq\lfloor\frac{n-k}{2}\rfloor-1$.  This implies $d-1\leq\lfloor \frac{n}{2} \rfloor-1-1$  as $k\geq 2$.  Hence, $d\leq\lfloor\frac{n}{2}\rfloor-1$ and the theorem follows.  \hfill \qed 
\end{proof} 
\subsection{A bound on the number of spanning trees in Halin graphs} 
We count the number of spanning trees of a Halin graph $H$ by constructing a computational tree of $H$ (see Figure \ref{ct}).  
\begin{figure}[!h] 
\begin{center} 
\includegraphics[scale=1]{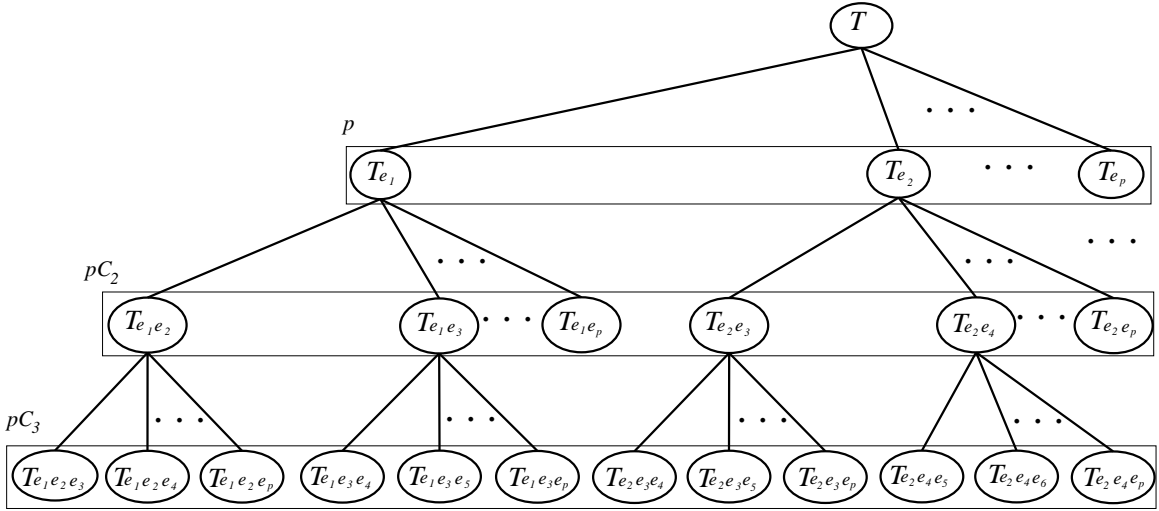}  
\caption{Computational Tree of H} \label{ct}
\end{center}
\end{figure}
At Level 1, the node $T_{e_i}$ represents a set that contains all spanning trees containing the edge $e_i\in C$.  At Level 3, any node with label $T_{e_i,e_j,e_k}$ represents a set containing all spanning trees that includes $e_i,e_j,e_k$ of $C$.  In general, at Level $l$, any node with label $T_{e_1,\ldots,e_i}$ represents a set containing all spanning trees that includes the edges $e_1,\ldots,e_i$ of $C$.  The next lemma gives an upper bound on the number of spanning trees in $T_{e_1,\ldots,e_i}$.  Note that ${e_1,\ldots,e_i}$ in $T_{e_1,\ldots,e_i}$ denotes some $i$ edges from $C$, need not be $i$ consecutive edges in $C$.
\begin{lemma} \label{maxsptreeati}
The number of spanning trees in $T_{e_1,\ldots,e_i}$ is $i!\cdot(2d)^i$.
\end{lemma}
\begin{proof}
We shall prove by induction on $i$. 
By our algorithm,  when $e_i$ is added to $T$, it creates a cycle of length at most $2d+1$.  The algorithm removes every cycle edge in $C^*$ except $e_i$ of $C$.  Therefore, for the base case, each cycle (accompanying cycle) edge when added to $T$ creates $2d$ spanning trees.  Let us consider the $i^{th}$ iteration $i\ge 2$.  Clearly, the algorithm adds the $i^{th}$ edge $e_i$ to all spanning trees generated in the previous iteration, $T_{e_1,\ldots,e_{i-1}}$.  That is, by the induction hypothesis, there are $(i-1)!\cdot(2d)^{i-1}$, $i\ge2$ spanning trees in $T_{e_1,\ldots,e_{i-1}}$ and to each said spanning trees, the cycle edge $e_i$ is added.
Consider a tree $T'$ in $T_{e_1,\ldots,e_{i-1}}$.  For any two consecutive edges $e_j, e_{j+1}$, $1\leq j<i-1$ of ${e_1,\ldots,e_{i-1}}$, we now bound the number of tree edges between $e_j$ and $e_{j+1}$.  It is clear that the number of such edges is at most $2d$.
When $e_i$ is added to a spanning tree in $T_{e_1,\ldots,e_{i-1}}$, the length of the cycle created is at most $i\cdot2d$.  This is true due to the above observation that there are $2d$ tree edges between any two consecutive cycle edges.
Therefore, the number of spanning trees generated in $T_{e_1,\ldots,e_i}$ is at most $(i-1)!\cdot(2d)^{i-1}\cdot i\cdot(2d)$ $=i!\cdot(2d)^{i}$.  This completes the proof. $\hfill\qed$ 
\end{proof}
\begin{theorem}
The number of spanning trees in any Halin graph is $O((2pd)^{p})$
\end{theorem}
\begin{proof}
Note that, in the computational tree of a Halin graph, there are $p$ nodes in Level 1.  Therefore, the number of spanning trees generated in Level 1 of the computational tree is at most $p\cdot2d$.  From Lemma \ref{maxsptreeati}, there exist $2!\cdot(2d)^{2}$ spanning trees in each node of Level 2 and there are $((p-1)+(p-2)+\ldots+1)$ nodes in Level 2.  Therefore, the number of spanning trees generated in Level 2 is at most $((p-1)+(p-2)+\ldots+1)\cdot2!\cdot(2d)^{2}$  $=\binom{p}{2}\cdot2!\cdot(2d)^{2}$.  There exist $\binom{p-1}{2}$ nodes in Level 3, which are descendents of $T_{e_1}$.  Similarly, there are $\binom{p-i}{2}$ nodes in Level 3, which are descendents of $T_{e_i}$, $1\leq i\leq p-2$.  Therefore, the number of nodes in Level 3 $=\binom{p-1}{2}+\binom{p-2}{2}+\ldots+\binom{2}{2}=\binom{p}{3}$.  In general, there are $\binom{p-i}{j}$ nodes, which are descendents of $T_{e_i}$ in Level $(j+1)$.  Hence there are $\Sigma_{i=1}^{p-j} \binom{p-i}{j}=\binom{p}{j+1}$ nodes in Level $(j+1)$.  Also from Lemma \ref{maxsptreeati}, there are at most $i!\cdot(2d)^{i}$ spanning trees in a node at Level $i$, $0\leq i\leq p-1$.  
%
In general, $i^{th}$ level has $\binom{p}{i}\cdot i!\cdot(2d)^{i}$ spanning trees.  Therefore, the total number of spanning trees is at most 
$$ \sum\limits_{i=1}^{p-1} \binom{p}{i}\cdot i!\cdot(2d)^i \le (p-1)!\cdot(1+2d)^{p-1} \le (p-1)!\cdot(2d)^p =O((2pd)^{p}) $$ $\hfill \qed$
\end{proof}
%
\section{Parallel algorithm: Listing all spanning trees in Halin graphs}
The overall structure of our sequential algorithm present in the previous section naturally gives a parallel algorithm to enumerate all spanning trees in Halin graphs.  Note that the length of the cycle $C'$ created due to the addition of edge $e\in C$ is $O(pd)$, and the last but one iteration of the algorithm generates at most $(p-2)!\cdot(2d)^{p-2}$ spanning trees.  Further, we know that the number of spanning trees generated is $O((2pd)^{p})$.  Using these facts, we have fixed the number of processors to be $O((2pd)^p)$ and our implementation is based on CRCW PRAM.  
%
%
\\\\ 
\begin{algorithm}[!h] 
\caption{Parallel algorithm to list all spanning trees of a Halin graph \newline {\em parallel-list-spanning-trees(H)}} 
\label{alg3}
\begin{algorithmic}[1]
\STATEx{{\tt Input:} A Halin Graph $H$}
\STATEx{{\tt Output:} All spanning trees of $H$. }
\STATEx{\tt /* The set $ENUM$ contains all spanning trees of $H$ and $T$ is the characteristic tree of $H$.   */ }
\STATE{ Initialize  \emph{ENUM} = $\{T\}$.}
\STATE{$ \sigma $=$(e_1,\ldots,e_p)$  be an ordering of edges in C.}
\STATE{\textbf{cobegin}}
\FOR{$i$ $=$ $1$ to $p$}
\STATE{Assign $(T,e_i)$ to a new distinct processor $P$.}
\STATE{$P$ calls {\tt parallel-recursive-list(T,$e_i$)}.}
\ENDFOR
\STATE{\textbf{coend}}
\end{algorithmic}
\end{algorithm}
\begin{algorithm}[!h]  
\caption{Parallel recursive listing of spanning trees in Halin graphs \newline {\em parallel-recursive-list($T^{'}$,$e_i$)}}
\label{alg4}
\begin{algorithmic}[1] 
\STATE{$G^*\leftarrow$ $T^{'}$+ $e_i$.}  
\STATE{Let $C^{*}$ be the unique cycle containing $e_i$ .}
\STATE{$ \sigma^{*} $=$(b_1,\ldots,b_k)$  be an ordering of edges in $E(C^{*})\backslash E(C)$.}
\STATE{\textbf{cobegin}}
\FOR{$m$ $=$ $1$ to $k$}
\STATE{$T^{*}\leftarrow $ $G^{*}-$ $b_m$. \emph{ENUM} = \emph{ENUM} $\cup$ $T^{*}$.}
\STATE{\textbf{cobegin}}
\FOR{$j$ $=$ $i+1$ to $p$} 
\STATE{Assign $(T^{*},e_j)$ to a new distinct processor P. i.e., execute {\tt parallel-recursive-list($T^{*}$,$e_j$)}.}
\ENDFOR
\STATE{\textbf{coend}}
\ENDFOR
\STATE{\textbf{coend}}
\end{algorithmic}
\end{algorithm}
\subsection{Proof of Correctness}
\begin{theorem}
For a Halin Graph $H,$ Algorithm 3 enumerates all spanning trees of $H$.
\end{theorem}
\begin{proof}
We present the proof of correctness of Algorithm 3 using the technique proof by minimum counter example. Let $R$ represents the set of spanning trees of $H$ which are not enumerated by Algorithm 3 and $T^{'}\in R$ is a spanning tree with a minimum number of accompanying cycle edges.\\ 
{\em Case 1:} $T^{'}$ is not the characteristic tree of $H$ as Algorithm 3 generates the characteristic tree $T$ in the first step itself.\\
{\em Case 2:} $E(T^{'})\cap E(C)\neq\emptyset$.  Let $E(T^{'})\cap E(C)$=$\{e_1,\ldots,e_k\}$. Step 5 of Algorithm 3 assigns $(T,e_1)$ to a new distinct processor P which in turn calls Algorithm 4 with the parameters as $(T,e_1)$.  Algorithm 4 adds $e_1$ to $T$ and enumerates all possible spanning trees containing $e_{1}$. 
In the similar way, Algorithm 4 adds $e_2, \ldots, e_k$ recursively and enumerates all spanning trees which contain $\{e_1, \ldots, e_k\}$.  Therefore $T^{'}\notin R$ and $R$ is empty, and the theorem follows. \hfill \qed
\end{proof}
\subsection{Run-time Analysis}
In this section, we analyze the time taken by a single processor and we make use of $O((2pd)^p)$ processors.  Since the total number of spanning trees in a Halin graph is bounded by $O((2pd)^p)$, we use $O((2pd)^p)$  processors.  It is clear from the computational tree that the last but one iteration makes use of $O((2pd)^p)$ distinct processors.  Since there is a trade off between the number of processors and the run time, by using  $O((2pd)^p)$ processors, we achieve maximum parallelism in our enumeration approach.
The {\em parallel-list-spanning-trees(H)} is invoked on a processor $P_0$, which in turn invokes $p$ other new processors with the sub routine {\em parallel-recursive-list($T$,$e_i$)}, $1\leq i\leq p$.  Each newly activated processor $P_i, 1\le i\le p$ upon receiving the characteristic tree $T$ and an accompanying  cycle edge $e_i$, $1\leq i\leq p$ adds $e_i$ to $T$.  This creates a unique cycle $C^*$, which can be detected in linear time.  In the algorithm, $\sigma^*$ contains the possible edges that could be removed from $C^*$ to obtain new spanning trees.  Since there are $O(pd)$ edges in $\sigma^*$, $O(pd)$ spanning trees are spawned by each processor $P_i$.  Further, for every newly created spanning tree $T^{'},$ {\em parallel-recursive-list($T^{'},$ $e_j$)}, $i<j\leq p$ is spawned on a new processor.  Hence, at most $O(p^2 d)$ processors are activated by any processor $P_i$, each can be done in parallel.  The same analysis holds good for all newly spawned processors.  Therefore, we conclude that each of the processor needs $O(p^2 d)$ effort other than $P_0$, which incurs $O(p)$ time.  Therefore, the overall effort of each processor in this parallel approach is $O(p^2 d)$.
\section{Parallel Algorithm: Listing all spanning trees in Halin graphs without repetitions}
In the previous section, we have shown that our sequential and parallel algorithm (Algorithms 1 and 3) lists all spanning trees of a Halin graph.  It is important to highlight that our algorithm enumerates with repetitions.  In this section, we shall present a parallel algorithm for listing all spanning trees in a Halin graph without repetitions.  Since our parallel algorithm is a natural extension of the associated sequential algorithm, the associated sequential algorithm also enumerates without repetitions.  For succint presentation, we present only the parallel version and avoid presenting the sequential version.  \\\\
We introduce a \emph{coloring of edges} during enumeration.  
Similar to Algorithm 3, Algorithm 5 has an ordering for the edges of $E(C)$ that are added  to $T$.  For each edge of $E(C)$ when it is added to the tree under consideration, a cycle is created and to obtain another spanning tree there exists more than one choice for edges that can be removed in Algorithm 3.  Therefore, the same spanning tree is created more than once by different ordering of edge deletions.  For example, in Figure \ref{trace}, the spanning tree $12a\_2$ is obtained by deleting edges $v_1v_3, v_1v_2$ in order, whereas the identical spanning tree $12b\_2$ is obtained by deleting edges $v_1v_2, v_1v_3$ in order.\\\\
Using our coloring scheme in Algorithm 5, we are ensuring the existence of only one ordering for edge deletions.  To explain the details of coloring, we shall see some notation as follows.   Since Halin graphs are planar, we work with the underlying plane embedding and with respect to this embedding, we order the edges in the accompanying cycle as $\sigma =(e_1,e_2,\ldots,e_p)$.  Let $V(C)=\{v_1,v_2,\ldots,v_p\}$ such that $e_i=v_iv_{i+1}$, $1\leq i\leq p-1$ and $e_{p}=v_pv_1$.  Note that the characteristic tree $T$ is a rooted tree (say rooted at vertex $u$).  During the course of our proposed algorithm, we add an edge $e_i=v_iv_{i+1}\in E(C)$ to a spanning tree $T'$, which creates a cycle $C^*$.  It is easy to see that there exist unique paths $P_{uv_i}$ and $P_{uv_{i+1}}$ in $T'$.  Now we shall define some more paths with respect to $T'$.  We identify a vertex $v$ in $P_{uv_i}\cap P_{uv_{i+1}}$ such that $P_{vv_i}$ and $P_{vv_{i+1}}$ are two vertex disjoint paths, $P_{vv_i}\sqsubseteq P_{uv_i}$ and $P_{vv_{i+1}}\sqsubseteq P_{uv_{i+1}}$. 
%
We define $E^R_{e_i}=\{e:e\in E(P_{vv_{i+1}}) \}$, and $E^L_{e_i}=E(C^{*})\backslash (B\cup E^R_{e_i}\cup C)$ where $B$ is the set of colored edges in $C^*$.  Now it is easy to observe that $E^R_{e_i}$ and $ E^L_{e_i}$ forms a partition of $E(C^*)\backslash (B\cup C)$, and all edges in $E^R_{e_i}$ are uncolored.  With respect to an edge $f=xw\in E^R_{e_i}$, we define $E^{f}_{e_i}=E(P_{xv_{i+1}})$, where $P_{xv_{i+1}}\sqsubseteq P_{wv_{i+1}}$ and  $P_{wv_{i+1}}\sqsubseteq P_{vv_{i+1}}$.  \\\\
Coloring of edges is done when we delete an uncolored edge $f$ from $E(C^*)\backslash E(C)$ ( Recall that $C^*$ is created due to the addition of the edge $e_i$ ).  That is, if the uncolored edge $f\in E^L_{e_i}$, then all edges in $E^R_{e_i}$ are colored \emph{blue}.  On the other hand if the uncolored edge $f\in E^R_{e_i}$, then color all edges of $E^f_{e_i}$ blue. 
\begin{figure}
\centering
\includegraphics[scale=1]{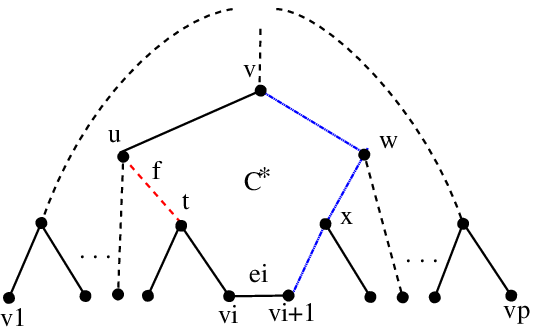}
\includegraphics[scale=1]{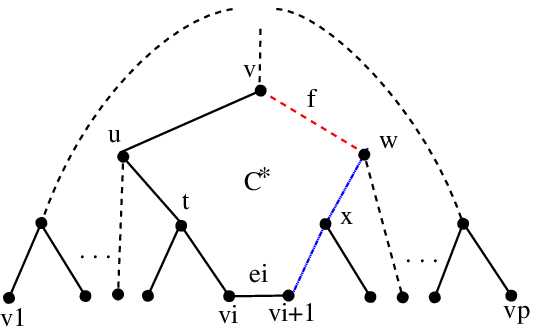}
\caption{An illustration on coloring of edges}
\label{figcoloring}
\end{figure}
%
See Figure \ref{figcoloring}, where on the removal of the edge $f$, the edges $\{vw,wx,xv_{i+1}\}$ are colored blue in the first one, whereas the edges $\{wx,xv_{i+1}\}$ are colored in the later.  The algorithm for enumeration without repetitions is similar to the previous parallel algorithm, with some constraints on the choice of edge for deletion.  That is, when we add an edge $e_i$ to a spanning tree $T'$ which creates a cycle $C^*$ containing $e_i$, we delete an uncolored edge $f\in E^R_{e_i}\cup E^L_{e_i}$.  Deletion of such an edge $f$ creates a spanning tree $T''$ with some more edges colored blue, and colored edges are not removed in the further enumeration process.  We highlight the fact that the coloring information is kept in the memory local to the processor.  It is due to this coloring scheme that the enumeration is done without repetitions.  
\begin{algorithm}[!h]
\caption{Parallel algorithm to list all distinct spanning trees of a Halin graph\newline {\em parallel-list-distinct-spanning-trees($H$)}} \label{algnorep}
\begin{algorithmic}[1]
\STATEx{{\tt Input:} A Halin Graph $H$}
\STATEx{{\tt Output:} All distinct spanning trees of H. \newline {\tt /* The set $ENUM$ contains all spanning trees of $H$ and $T$ is the characteristic tree of $H$ rooted at $u$ */ }}
\STATE{ Initialize $ENUM=\{T\}$.}
\STATE{$ \sigma $=$(e_1,\ldots,e_p)$  be an ordering of edges in C.}
\STATE{\textbf{cobegin}}
\FOR{$i$ $=$ $1$ to $p$}
\STATE{Assign $(T,e_i)$ to a new distinct processor $P$.}
\STATE{$P$ calls {\tt parallel-recursive-distinct-list(T,$e_i$)}.}
\ENDFOR
\STATE{\textbf{coend}}
\end{algorithmic}
\end{algorithm}
\begin{algorithm}[!h]
\caption{Parallel recursive listing of distinct spanning trees in Halin graphs \newline {\em parallel-recursive-distinct-list($T^{'}$,$e_i$)}}
\begin{algorithmic}[1]
\STATE{$G^*\leftarrow$ $T^{'}$+ $e_i$.  Find $C^*$, the cycle in $G^*$ containing $e_i$.}   
\STATE{$ \sigma^{*}_1 $=$(b_1,\ldots,b_j)$ be an ordering of edges in $E^L_{e_i}$ and \newline$\sigma^{*}_2$=$(b_{j+1},\ldots,b_k)$ be an ordering of edges in $E^R_{e_i}$.}
\STATE{\textbf{cobegin}}
\FOR{$m$ $=$ $1$ to $|C^*|$}
\IF{$b_m$ is uncolored}
\IF{$b_m\in\sigma^{*}_1$}
\STATE{$T^{*}\leftarrow $ $G^{*}-$ $b_m$}
\STATE{$\forall b_q\in \sigma^{*}_2$, color $b_q$ as blue.}
\ELSE
\STATE{$T^{*}\leftarrow $ $G^{*}-$ $b_m$}
\STATE{ $\forall b_q\in E^{b_m}_{e_i}$, color $b_q$ as blue.}
\ENDIF
\STATE{ENUM = ENUM $\cup$ $T^{*}$.}
\STATE{\textbf{cobegin}}
\FOR{$l$ $=$ $i+1$ to $p$} 
\STATE{Assign $(T^{*},e_l)$ to a new distinct processor.}
\STATE{ P calls {\tt parallel-recursive-distinct-list($T^{*}$,$e_l$)}.}
\ENDFOR
\STATE{\textbf{coend}}
\ENDIF
\ENDFOR
\STATE{\textbf{coend}}
\end{algorithmic}
\end{algorithm}
\begin{theorem} \label{thm4}
For a Halin Graph $H,$ Algorithm 5 generates all spanning trees of $H$.
\end{theorem}
\begin{proof}  
Let $\{f_1,\ldots,f_k\}\subseteq E(C)$ be the set of accompanying cycle edges in the spanning trees belonging to $T_{f_1,\ldots,f_k}$.  From the computational tree, it is clear that there exist a node $T_{f_1,\ldots,f_k}$ for every $\{f_1,\ldots,f_k\}\subseteq E(C)$, $1\leq k\leq p-1$.  To show that Algorithm 5 generates all spanning trees of $H$, it is sufficient to show that $T_{f_1,\ldots,f_k}$ contains all spanning trees having accompanying cycle edges $\{f_1,\ldots,f_k\}$.  We prove this using mathematical induction on $k$\\
{\em Base case:} When $k=1$, $T_{f_1}$ contains all spanning trees obtained by adding $f_1\in E(C)$ to $T$ and removing each edge from $E^L_{f_1}$ and $E^R_{f_1}$. \\
{\em Hypothesis:} $T_{f_1,\ldots,f_{k-1}}$, on at most $k-1$, $k\geq 2$ accompanying cycle edges has all spanning trees containing edges $\{f_1,\ldots,f_{k-1}\}$.\\
{\em Induction Step:}  Let $T_{1}\in T_{f_1,\ldots,f_k}$ be an arbitrary spannning tree on $k$, $k\geq 2$ accompanying cycle edges.   Let us consider the set of edges $M=E(T)\backslash E(T_{1})$ and $M=\{m_1,\ldots,m_k\}$.  Observe $|M|=|\{f_1,\ldots,f_k\}|$.  Notice that there exist at least one edge, say $m_i\in M$ $1\leq i\leq k$ such that $m_i$ when added to $T_1$ creates a cycle $C^*$ where $f_k\in E(C^*)$.   Clearly, $T_{2}$ where $E(T_2)=(E(T_1) \cup \{m_i\}) \backslash \{f_{k}\}$ is a spanning tree on $k-1$ accompanying cycle edges.  
%
 Therefore, by the induction hypothesis, $T_2$ is generated by Algorithm 5.  Note that, on the course of generation of $T_2$ by Algorithm \ref{algnorep}, the edge $m_i$ may be colored.  Therefore, we consider the following cases based on coloring of $m_i$.  \\
\textit{Case 1:} The edge $m_i$ is uncolored in $T_2$.  Then Algorithm 5 adds edge $f_k$, and removes $m_i$, which is the desired spanning tree $T_1$.\\
\textit{Case 2:} The edge $m_i$ is colored blue in $T_2$.  
Note that whenever a new edge $f_k$ is added, there does not exist a colored edge in $E^{R}_{f_k}$.  Moreover, there exist an edge $f_j\in \{f_1,\ldots,f_{k-1}\}$ such that the addition of the edge $f_j$ and the removal of the edge $m_h\in M$, $m_h\neq m_i$ colors the edge $m_i$.  It is clear that either $m_h\in E^{L}_{f_j}$, $m_i\in E^{R}_{f_j}$, and $m_i$ is colored blue, or $m_h, m_i\in E^{R}_{f_j}$,   $m_i\in E^{m_h}_{f_j}$, and $m_i$ is colored blue.  Now there exist a spanning tree $T_3\in T_{f_1,\ldots,f_{k-1}}$ such that $E(T_3)=(E(T_2)\backslash \{m_i\}) \cup \{m_h\}$, and $m_h$ remains uncolored.  Observe that when $f_k$ is added to all spanning trees in  $T_{f_1,\ldots,f_{k-1}}$, $f_k$ is also added to $T_3$ and $m_h$ is removed to obtain $T_1$.  i.e., $E(T_1)=(E(T_3)\backslash \{m_h\}) \cup \{f_k\}$.  This completes the induction and the proof of Theorem \ref{thm4}.
\hfill \qed
\end{proof}   
There are no two spanning trees $T_1\in T_{e_i,\ldots,e_j}$, $T_2\in T_{e_k,\ldots,e_l}$, $\{e_i,\ldots,e_j\}\neq \{e_k,\ldots,e_l\}$ such that $T_1=T_2$ since they differ in accompanying cycle edges.  Therefore, if there exists duplicate spanning trees generated by Algorithm \ref{algnorep}, then it is found only within a set $T_{e_i,\ldots,e_j}$.  Theorem \ref{thm5} shows that such duplicate trees are not produced by Algorithm \ref{algnorep}. 
\begin{obs}
\textit{ For a Halin graph $H=T\cup C$, let $e_i,e_j\in E(C)$ and $C_a, C_b$ be the cycles formed by the addition of edges $e_i, e_j$, respectively to $T$  then, $|E(C_a)\cap E(C_b)|\le1$.}
\end{obs}
\begin{theorem} \label{thm5}
For $\{f_1,\ldots,f_k\}\subset E(C)$, let $T_{f_1,\ldots,f_k}$ be a set of spanning trees generated by Algorithm 5 on an input Halin graph $H=T\cup C$.  There does not exist spanning trees $T_1,T_2\in$ $T_{f_1,\ldots ,f_k}$ such that $T_1=T_2$.  That is, there are no duplicate spanning trees in $T_{f_1,\ldots,f_k}$ 
\end{theorem}

\begin{proof}
Let $T_1\in$ $T_{f_1,\ldots,f_k}$  be a spanning tree and $M=E(T)\backslash E(T_1)$.  Assume for a contradiction that there exist a duplicate spanning tree $T_2$, such that $E(T_2)=E(T_1)$ enumerated by Algorithm \ref{algnorep}.  Since the algorithm adds edges $\{f_1,\ldots,f_k\}$ one after the other in order, edges in $M$ are removed in different order say $(a_1,\ldots,a_{k})$ to obtain $T_1$ and $(b_1,\ldots,b_{k})$ to obtain $T_2$ where $M=\{a_1,\ldots,a_k\}=\{b_1,\ldots,b_k\}$.  
Consider the least indexed edge $f_j\in \{f_1,\ldots,f_k\}$ such that $a_j\neq b_j$. 
Let $C^*$ be the cycle formed by the addition of the edge $f_j$.  Clearly, $a_j,b_j\in E(C^*)\cap M$ such that $a_j, b_j$ are removed first from $E(C^*)$ during the course of creation of  $T_1, T_2,$ respectively.  Note that since $a_j\in M$, $a_j\in \{b_{j+1},\ldots,b_k\}$.  That is,  $b_l=a_j,l>j$ is removed at a later point for obtaining $T_2$.  Now, note that $\{a_j,b_j\}\not \subset E^{L}_{f_j}$ since none of the cycles created by any of the edges in $\{a_{j+1},\ldots,a_k\}$ and $\{b_{j+1},\ldots,b_k\}$ could contain any of the edges $a_j$ and $b_j$.  Therefore, we come across the following two cases, and an illustration is given in Figure \ref{figthm5}.\\
\textit{Case 1:} $\{a_j,b_j\}\subset E^{R}_{f_j}$.  If $b_j\not \in E^{a_j}_{f_j}$, then for the spanning tree, $T_2$ while adding the edge $f_j$, the edge $b_j$ is removed, and $a_j$ is colored blue.  This contradicts the fact that $a_j\in E(T)\backslash E(T_2)$.  Now if $b_j\in E^{a_j}_{f_j}$, then for creating $T_1$, removal of the edge $a_j$ colors the edge $b_j$, thereby $b_j\notin E(T)\backslash E(T_1)$, a contradiction. 
\begin{figure}
\centering
\includegraphics[scale=1]{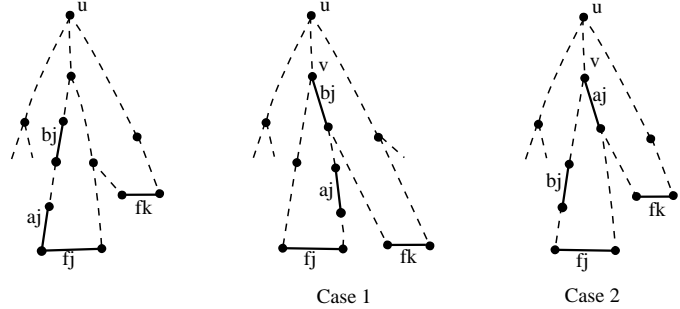}
\caption{An illustration for the proof of Theorem \ref{thm5}}
\label{figthm5}
\end{figure}\\
\noindent \textit{Case 2:}  $a_j\in E^{R}_{f_j}$ and  $b_j\in E^{L}_{f_j}$.  
For the spanning tree $T_2$ on adding the edge $f_j$, and removing the edge $b_j$, the edge $a_j$ is colored blue.  This contradicts the fact that $a_j\in E(T)\backslash E(T_2)$.  It is easy to see the symmetric case when $b_j\in E^{R}_{f_j}$ and  $a_j\in E^{L}_{f_j}$.  
Thus there exist a unique removable edge $a_i\in M$ corresponding to every added edge $f_i\in E(C^*)$, $1\leq i\leq k$ to obtain a spanning tree in Algorithm \ref{algnorep}.  This completes the proof of Theorem \ref{thm5}. \hfill \qed
\end{proof}
\begin{theorem}
For a Halin graph $H$, Algorithm \ref{algnorep} enumerates all spanning trees without repetitions.
\end{theorem}
\begin{proof}
Follows from Theorem \ref{thm4} and Theorem \ref{thm5}. \hfill \qed
\end{proof}
\subsection{Run-time Analysis}
Algorithm \ref{algnorep} is invoked on a processor $P_0$, which invokes $p$ other new processors with the sub routine {\em  parallel-recursive-distinct-list($T$,$e_i$)}, $1\leq i\leq p$.  Each newly activated processor $P_i$  when adds $e_i$ to $T$, a cycle $C^*$ is formed, which is detected in $O(pd)$ time.  There are at most $O(pd)$ uncolored edges in $\sigma^*$, which can be removed one after the other from $C^*$ to obtain new  spanning trees.  At most $d$ edges in $E^R_{e_i}$ are colored for each of the newly obtained spanning trees.  Therefore, $P_i$ incurs $O(pd^2)$ effort in creating $O(pd)$ spanning trees.  For each newly created spanning tree $T^{'},$ {\em  parallel-recursive-distinct-list($T^{'},$ $e_j$)}, $i<j\leq p$ is spawned on a new processor.  $P_i$ activates, at most $O(p^2d)$ processors.  Therefore, the overall effort of each processor in this parallel approach is $O(pd^2+p^2d)$.  
\section{Conclusions and Future Work}
We have presented a sequential and parallel algorithm to list all spanning trees in Halin Graphs without repetitions.  We have also presented a bound on the number of spanning trees generated by our algorithm.  Interesting problems for further research are to enumerate all minimal vertex separators and maximal independent sets in Halin graphs and other special graph classes.   
\bibliographystyle{splncs1}
\bibliography{halinref}
%
\end{document}